\documentclass[a4paper,UKenglish,cleveref, autoref, thm-restate]{lipics-v2021}

\newif\ifFull
\Fullfalse

\usepackage{caption}
\usepackage{subcaption}
\usepackage{tikz}

\usepackage{pgfplots}
\usepackage{framed,enumitem} 
\usepackage{algorithm}
\usepackage{numprint}
\usepackage[noend]{algpseudocode}

\DeclareMathOperator*{\argmax}{argmax}
\DeclareMathOperator*{\argmin}{argmin}

\newcommand{\Is}       {:=}
\newcommand{\set}[1]{\left\{ #1\right\}}

\newcommand{\gilt}{:}
\newcommand{\sodass}{\,:\,}
\newcommand{\setGilt}[2]{\left\{ #1\sodass #2\right\}}

\def\MdR{\ensuremath{\mathbb{R}}}

\newcommand{\NEW}[1]{{\color{black} #1}}

\title{FREIGHT: Fast Streaming Hypergraph Partitioning} %

\author{Kamal {Eyubov}}{Heidelberg University, Germany }{kamal.eyubov@stud.uni-heidelberg.de}{kamal.eyubov@stud.uni-heidelberg.de}{}%

\author{Marcelo {Fonseca Faraj}}{Heidelberg University, Germany}{marcelofaraj@informatik.uni-heidelberg.de}{https://orcid.org/0000-0001-7100-236X}{}

\author{Christian {Schulz}}{Heidelberg University, Germany}{christian.schulz@informatik.uni-heidelberg.de}{https://orcid.org/0000-0002-2823-3506}{}

\authorrunning{Eyubov et al.} %

\Copyright{Jane Open Access and Joan R. Public} %



\category{} %

\relatedversion{} %

\EventEditors{John Q. Open and Joan R. Access}
\EventNoEds{2}
\EventLongTitle{21st Symposium on Experimental Algorithms (SEA 2023)}
\EventShortTitle{SEA 2023}
\EventAcronym{SEA}
\EventYear{2023}
\EventDate{July 24--26, 2023}
\EventLocation{Barcelona, Spain}
\EventLogo{}
\SeriesVolume{42}
\ArticleNo{23}

\nolinenumbers 
\hideLIPIcs

\begin{document}

\maketitle

\begin{abstract}
Partitioning the vertices of a (hyper)graph into $k$ roughly balanced blocks such that few (hyper)edges run between blocks is a key problem for large-scale distributed processing. 
A current trend for partitioning huge (hyper)graphs using low computational resources are streaming algorithms. 
In this work, we propose FREIGHT: a Fast stREamInG Hypergraph parTitioning algorithm which is an adaptation of the widely-known graph-based algorithm Fennel.
By using an efficient data structure, we make the overall running of FREIGHT linearly dependent on the pin-count of the hypergraph and the memory consumption linearly dependent on the numbers of nets and blocks.
The results of our extensive experimentation showcase the promising performance of FREIGHT as a highly efficient and effective solution for streaming hypergraph partitioning. 
Our algorithm demonstrates competitive running time with the Hashing  algorithm, with a difference of a maximum factor of four observed on three fourths of the instances.
Significantly, our findings highlight the superiority of FREIGHT over all existing (buffered) streaming algorithms and even the in-memory algorithm HYPE, with respect to both cut-net and connectivity measures. 
This indicates that our proposed algorithm is a promising hypergraph partitioning tool to tackle the challenge posed by large-scale and dynamic data processing.
\end{abstract}


\section{Introduction}
\label{sec:introduction}

Graphs are ubiquitous in nature and can be used to represent a wide variety of phenomena such as road networks, dependencies in databases, communications in distributed algorithms, interactions in social networks, and so forth.
Nevertheless, phenomena where interactions between entities are not necessarily pairwise are more adequately modeled by hypergraphs, which can capture higher-order interactions~\cite{lambiotte2019networks}.
With the massive proliferation of data, processing large-scale (hyper)graphs on distributed systems and databases becomes a necessity for a wide range of applications.
When processing a (hyper)graph in parallel, $k$ processors operate on distinct portion of the (hyper)graph while communicating to one another through message-passing.
To make the parallel processing efficient, an important preprocessing step consists of partitioning the vertices of the (hyper)graph into $k$ roughly balanced blocks such that few (hyper)edges run between blocks.
(Hyper)graph partitioning is NP-hard \cite{Garey1974} and there can be no approximation algorithm with a constant ratio for general (hyper)graphs~\cite{BuiJ92}. 
Thus, heuristics are used in practice.
A current trend for partitioning huge (hyper)graphs quickly and using \hbox{low computational resources are streaming algorithms~\cite{tsourakakis2014fennel,awadelkarim2020prioritized,jafari2021fast,HeiStream,StreamMultiSection,mayer2018adwise,hoang2019cusp,alistarh2015streaming,tacsyaran2021streaming}}.

The most popular streaming approach in literature is the one-pass model~\cite{DBLP:journals/pvldb/AbbasKCV18}, where vertices arrive one at a time including their (hyper)edges and then have to be permanently assigned to blocks. 
In the domain of graphs, most algorithms are either very fast but do not care for solution quality at all (such as \texttt{Hashing}~\cite{stanton2012streaming}), or are still fast, but much slower and capable of computing significantly better solutions than just random assignments (such as such  \texttt{Fennel}~\cite{tsourakakis2014fennel}).
Recently, the gap between these groups of algorithms has been closed by a streaming multi-section algorithm~\cite{StreamMultiSection} which is up to two orders of magnitude faster than \texttt{Fennel} while cutting only $5\%$ more edges than in on average.
In the domain of hypergraphs, there is a similar gap that has not yet been closed.
In particular, there is the same trivial \texttt{Hashing} -based algorithm on one side, and more sophisticated and \hbox{expensive algorithms~\cite{alistarh2015streaming,tacsyaran2021streaming} on~the~other~side}.

In this work, we propose \texttt{FREIGHT}: a Fast stREamInG Hypergraph parTitioning algorithm that can optimize for the cut-net as well as the connectivity metric. 
By using an efficient data structure, we make the overall running of \texttt{FREIGHT} linearly dependent on the pin-count of the hypergraph and the memory consumption linearly dependent on the numbers of nets and blocks.
Our proposed algorithm demonstrates remarkable efficiency, with a running time comparable to the \texttt{Hashing}  algorithm and a maximum discrepancy of only four in three quarters of the instances. 
Importantly, our study establishes the superiority of \texttt{FREIGHT} over all current (buffered) streaming algorithms and even the in-memory algorithm \texttt{HYPE}, in both cut-net and connectivity measures. 
This shows the potential of our algorithm as a valuable tool for partitioning hypergraphs in the context of \hbox{large and constantly changing data processing environments}.

\section{Preliminaries}
\label{sec:preliminaries}

\subsection{Basic Concepts}
\label{subsec:basic_concepts}

\subparagraph*{Hypergraphs and Graphs.}
Let $H=(V=\{0,\ldots, n-1\},E)$ be an \emph{undirected hypergraph} with no multiple or self hyperedges, with $n = |V|$ vertices and $m = |E|$ hyperedges (or \emph{nets}).
A net is defined as a subset of $V$. 
The vertices that compose a net are called \emph{pins}.  
A vertex $v\in V$ is \emph{incident} to a net $e\in E $ if $v \in e$.  
Let $c: V \to \MdR_{\geq 0}$ be a vertex-weight function, and let $\omega: E \to \MdR_{>0}$ be a net-weight function.
We generalize $c$ and $\omega$ functions to sets, such that $c(V') = \sum_{v\in V'}c(v)$ and $\omega(E') = \sum_{e\in E'}\omega(e)$.
Let $I(v)$ be the set of incident nets of $v$, let $d(v) \Is |I(v)|$ be the \emph{degree} of $v$, let $d_{w}(v) \Is w(I(v))$ be the \emph{weighted degree} of $v$, and let $\Delta$ be the maximum degree of $H$.
We generalize the notations $d(.)$ and  $d_{w}(.)$ to sets, such that $d(V') = \sum_{v\in V'}d(v)$ and $d_{w}(V') = \sum_{v\in V'}d_{w}(v)$.
Two vertices are \emph{adjacent} if they are incident to a same net.                      
Let the number of pins~$|e|$ in a net~$e$ be the \emph{size} of~$e$, let $\xi = \max_{e \in E}\{|e|\}$ be the maximum size of a net in $H$.

Let $G=(V=\{0,\ldots, n-1\},E)$ be an \emph{undirected graph} with no multiple or self edges, such that $n = |V|$, $m = |E|$.
Let $c: V \to \MdR_{\geq 0}$ be a vertex-weight function, and let $\omega: E \to \MdR_{>0}$ be an edge-weight function.
We generalize $c$ and $\omega$ functions to sets, such that $c(V') = \sum_{v\in V'}c(v)$ and $\omega(E') = \sum_{e\in E'}\omega(e)$.
Let $N(v) = \setGilt{u}{\set{v,u}\in E}$ denote the neighbors of $v$.
A graph $S=(V', E')$ is said to be a \emph{subgraph} of $G=(V, E)$ if $V' \subseteq V$ and $E' \subseteq E \cap (V' \times V')$. 
When $E' = E \cap (V' \times V')$, $S$ is an \emph{induced} subgraph.
Let $d(v)$ be the degree of vertex $v$ and $\Delta$ be the maximum degree of $G$.%

\subparagraph*{Partitioning.}
The \emph{(hyper)graph partitioning} problem consists of assigning each vertex of a (hyper)graph to exactly one of $k$ distinct \emph{blocks} respecting a balancing constraint in order to minimize the weight of the (hyper)edges running between the blocks, i.e., the edge-cut (resp. cut-net).
More precisely, it partitions $V$ into $k$ blocks $V_1$,\ldots,$V_k$ (i.e., $V_1\cup\cdots\cup V_k=V$ and $V_i\cap V_j=\emptyset$ for $i\neq j$), which is called a \emph{\mbox{$k$-partition}} of the (hyper)graph.
The \emph{edge-cut} (resp. \emph{cut-net}) of a $k$-partition consists of the total weight of the \emph{cut edges} (resp. \emph{cut nets}), i.e., edges (resp. nets) crossing blocks.
More formally, let the edge-cut (resp. cut-net) be $\sum_{i<j}\omega(E')$, in which $E' \Is $ $\big\{e\in E, \exists \set{u,v} \subseteq e : u\in V_i,v\in V_j, i \neq j\big\}$ is the~\emph{cut-set} (i.e.,~the set of all cut nets).
The \emph{balancing constraint} demands that the sum of vertex weights in each block does not exceed a threshold associated with some allowed \emph{imbalance}~$\epsilon$.
More specifically, $\forall i~\in~\{1,\ldots,k\} \gilt$ $c(V_i)\leq L_{\max}\Is \big\lceil(1+\epsilon) \frac{c(V)}{k} \big\rceil$.
For each net $e$ of a hypergraph, $\Lambda(e) := \{V_i~|~V_i \cap e \neq \emptyset\}$ denotes the \emph{connectivity set} of $e$.
The \emph{connectivity} $\lambda(e)$ of a net~$e$ is the cardinality of its connectivity set, i.e., $\lambda(e) := |\Lambda(e)|$.
The so-called \emph{connectivity} metric ($\lambda$-1) is computed as $\sum_{e\in E'} (\lambda(e) -1)~\omega(e)$, where $E'$ is the cut-set.

\subparagraph*{Streaming.}
\label{subsec:Computational Model}
Streaming algorithms usually follow an iterative load-compute-store logic.
Our focus and the most used streaming model is the \emph{one-pass} model.
In this model, vertices of a (hyper)graph are loaded one at a time alongside with their adjacency lists, then some logic is applied to permanently assign them to blocks, as illustrated in Figure~\ref{fig:StreamingModel}.
A similar sequence of operations is used to partition a stream of edges of a graph on the fly.
In this case, edges of a graph are loaded one at a time alongside with their end-points, then some logic is applied to permanently assign them to blocks.
This logic can be as simple as a \texttt{Hashing}  function or as complex as scoring all blocks based on some objective and then assigning the vertex to the block with highest score.
There are other, more sophisticated, streaming models such as the sliding window~\cite{patwary2019window} and the buffered streaming~\cite{jafari2021fast,HeiStream}, but are \hbox{beyond the scope of this work}.

\begin{figure}[t]
	\centering
	\includegraphics[width=0.8\linewidth]{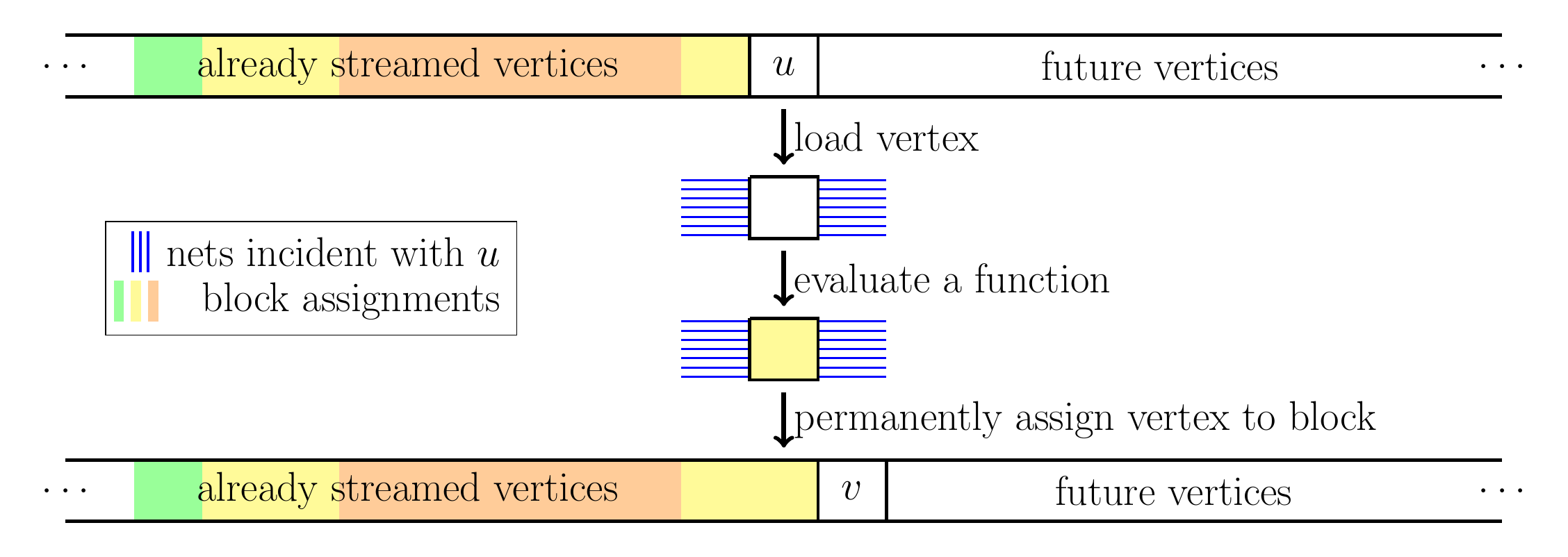}
	\caption{Typical layout of streaming algorithm for hypergraph partitioning.}
	\label{fig:StreamingModel}
\end{figure}

\subsection{Related Work}
\label{subsec:related_work}

There is a huge body of research on (hyper)graph partitioning.
The most prominent tools to partition (hyper)graphs in memory include
\texttt{PaToH}~\cite{ccatalyurek2011patoh},  \texttt{Metis}~\cite{karypis1996parallel},  \texttt{hMetis}~\cite{hMetis},  \texttt{Scotch}~\cite{Pellegrini96experimentalanalysis}, \texttt{HYPE}~\cite{HYPE2018},  \texttt{KaHIP}~\cite{kabapeE},  \texttt{KaMinPar}~\cite{gottesburen2021deep},  \texttt{KaHyPar}~\cite{schlag2016k},  \texttt{Mt-KaHyPar}~\cite{mt-kahypar-d}, and  \texttt{mt-KaHIP}~\cite{DBLP:conf/europar/Akhremtsev0018}.
The readers are referred to~\cite{more_recent_advances_hgp,SPPGPOverviewPaper,DBLP:reference/bdt/0003S19} for extensive material and references.
Here, we focus on the results specifically related to the scope of this paper.
In particular, we provide a detailed review for the following problems based on the one-pass streaming model: hypergraph partitioning and graph vertex partitioning.

\subparagraph*{Streaming Hypergraph Partitioning.}
Alistarh~et~al.~\cite{alistarh2015streaming} propose  \texttt{Min-Max}, a one-pass streaming
algorithm to assign the vertices of a hypergraph to blocks.
For each block, this algorithm keeps track of nets which contain pins in it.
This implies a memory consumption of $O(mk)$.
When a vertex is loaded, \texttt{Min-Max} allocates it to the block containing the largest intersection with its nets while respecting a hard constraint for load balance.
The authors theoretically prove that their algorithm is able to recover a hidden \emph{co-clustering}
with high probability, where a co-clustering is defined as a simultaneous clustering of vertices and hyperedges.
In the experimental evaluation, \texttt{Min-Max} outperforms five intuitive streaming approaches with respect to load imbalance, while producing solutions up to five times more imbalanced than internal-memory algorithms such as  \texttt{hMetis}.

Ta{\c{s}}yaran~et~al.~\cite{tacsyaran2021streaming} propose \NEW{improved versions of} the algorithm  \texttt{Min-Max}~\cite{alistarh2015streaming}.
The authors present \NEW{\texttt{Min-Max-N2P},} a modified version of \texttt{Min-Max} \NEW{that stores blocks containing each net's pins instead of storing nets per block, as done in \texttt{Min-Max}}.
In their experiments, \NEW{\texttt{Min-Max-N2P}} is three orders of magnitude faster than \texttt{Min-Max} while keeping the same cut-net. %
The authors also introduce three \NEW{algorithms with reduced memory usage compared to \texttt{Min-Max}: \texttt{Min-Max-L$\ell$}, a modification of \texttt{Min-Max-N2P} that employs an upper-bound $\ell$ to limit memory consumption per net, \texttt{Min-Max-BF} which utilizes Bloom filters for membership queries, and \texttt{Min-Max-MH} that uses hashing functions to replace the connectivity information between blocks and nets}.
In their experiments, their three algorithms reduce the running time in comparison to  \texttt{Min-Max}, especially \NEW{\texttt{Min-Max-L$\ell$}} and \NEW{\texttt{Min-Max-MH}}, which are up to four orders of magnitude faster.
On the other hand, the three algorithms generate solutions with worse cut-net than  \texttt{Min-Max}, especially \NEW{\texttt{Min-Max-MH}}, which increases the cut-net by up to an order of magnitude.
Moreover, the authors propose a \NEW{technique to improve the partitioning decision in the streaming setting by including a buffer to store some vertices and their net sets.} 
\NEW{This approach operates similarly to \texttt{Min-Max-N2P}, but with the added ability to revisit buffered vertices and adjust their partition assignment based on the connectivity metric.}
\NEW{The authors propose three algorithms using this buffered approach: \texttt{REF} that buffers every incoming vertex but only reassigns those that may improve connectivity, \texttt{REF\_RLX} that buffers all vertices and reassigns all vertices in the buffer, and \texttt{REF\_RLX\_SV} that only buffers vertices with small net sets and reassigns all vertices in the buffer.}
Their experimental results show that the \NEW{use of buffered approaches leads to a $5$-$20\%$ improvement in partitioning quality compared to non-buffered approaches, but with a trade-off of increased runtime}.

\subparagraph*{Streaming Graph Vertex Partitioning.}
Stanton and Kliot~\cite{stanton2012streaming} introduced graph partitioning in the streaming model and proposed some heuristics to solve it.
Their most prominent heuristic include the one-pass methods \texttt{Hashing} and \emph{linear deterministic greedy}~(\texttt{LDG}).
In their experiments, \texttt{LDG} had the best overall edge-cut.
In this algorithm, vertex assignments prioritize blocks containing more neighbors and use a penalty multiplier to control imbalance. 
Particularly, a vertex $v$ is assigned to the block $V_i$ that maximizes $|V_i \cap N(v)|*\lambda(i)$ with $\lambda(i)$ being a multiplicative penalty defined as $(1-\frac{|V_i|}{L_\text{max}})$. 
The intuition is that the penalty avoids to overload blocks that are already very heavy.
In case of ties on the objective function, \texttt{LDG} moves the vertex to the block with fewer vertices.
Overall, \texttt{LDG} partitions a graph in $O(m+nk)$ time.
On the other hand, \texttt{Hashing}  has running time $O(n)$ \hbox{but produces a poor edge-cut}.

Tsourakakis~et~al.~\cite{tsourakakis2014fennel} proposed \texttt{Fennel},  a one-pass partitioning heuristic based on the widely-known clustering objective \emph{modularity}~\cite{brandes2007modularity}.
\texttt{Fennel} assigns a vertex $v$ to a block $V_i$, respecting a balancing threshold, in order to maximize an expression of type $|V_i\cap N(v)|-f(|V_i|)$, i.e., with an additive penalty.
This expression is an interpolation of two properties: attraction to blocks with many neighbors and repulsion from blocks with many non-neighbors.
When $f(|V_i|)$ is a constant, the expression coincides with the first property.
If $f(|V_i|) = |V_i|$, the expression coincides with the second property.
In particular, the authors defined the  \texttt{Fennel} objective with $f(|V_i|) = \alpha * \gamma * |V_i|^{\gamma-1}$, in which~$\gamma$ is a free parameter and $\alpha = m \frac{k^{\gamma-1}}{n^{\gamma}}$.
After a parameter tuning made by the authors,  \texttt{Fennel} uses $\gamma=\frac{3}{2}$, which provides $\alpha=\sqrt{k}\frac{m}{n^{3/2}}$.
As \texttt{LDG},  \texttt{Fennel} \hbox{partitions a graph in $O(m+nk)$ time}.

Faraj~and~Schulz~\cite{StreamMultiSection} propose a shared-memory streaming algorithm for vertex partitioning which performs recursive multi-sections on the fly. 
As a preliminary phase, their algorithm decomposes a $k$-way partitioning problem into a hierarchy containing $\lceil\log_b k\rceil$ layers of $b$-way partitioning subproblems.
This hierarchy can either reflect the topology of a high performance system to solve a process mapping~\cite{HighQualityHierarchicalPM20,PredariProcessMapping21} or be computed for an arbitrary $k$ to solve a regular vertex partitioning.
Then, an adapted version of \texttt{Fennel} is used to solve each of the subproblems in such a way that the whole $k$-partition is computed on the fly during a single pass over the graph.
While producing an edge-cut around $5\%$ lower than \texttt{Fennel}, their algorithm has theoretical complexity $O((m+nb)\log_b k)$ and experimentally ran up to two \hbox{orders of magnitude faster than \texttt{Fennel}}.

Besides the one-pass model, other streaming models have also been used to solve vertex partitioning.
Restreaming graph partitioning has been introduced by Nishimura~and~Ugander~\cite{nishimura2013restreaming}.
In this model, multiple passes through the entire input are allowed, which enables iterative improvements.
The authors proposed easily implementable restreaming versions of \texttt{LDG} and \texttt{Fennel}: \texttt{ReLDG} and \texttt{ReFennel}, respectively.
Awadelkarim and Ugander~\cite{awadelkarim2020prioritized} studied the effect of vertex ordering for streaming graph partitioning.
The authors introduced the notion of \emph{prioritized streaming}, in which (re)streamed vertices are statically or dynamically reordered based on some priority. 
The authors proposed a prioritized version of \texttt{ReLDG}
Patwary et al.~\cite{patwary2019window} proposed \texttt{WStream}, a greedy stream algorithm that keeps a sliding stream window.
Jafari et al.~\cite{jafari2021fast} proposed a shared-memory multilevel algorithm based on a buffered streaming model.
Their algorithm uses the one-pass algorithm \texttt{LDG} as the coarsening, initial partitioning, and the local search steps of their multilevel scheme.
Faraj~and~Schulz~\cite{HeiStream} proposed \texttt{HeiStream}, a multilevel algorithm also based on a buffered streaming model.
Their algorithm loads a chunk of vertices, builds a model, and then partitions this model with a traditional multilevel algorithm coupled with an \hbox{extended version of the \texttt{Fennel} objective}.

\section{FREIGHT: Fast Streaming Hypergraph Partitioning}
\label{sec:FREIGHT: Fast Streaming Hypergraph Partitioning}

In this section, we provide a detailed explanation of our algorithmic contribution.
First, we define our algorithm named \texttt{FREIGHT}.
Next, we present the advantages and disadvantages of using two different formats for streaming hypergraphs and partitioning them using \texttt{FREIGHT}.
Additionally, we explain how we have removed the dependency on $k$ from the complexity of \texttt{FREIGHT} by implementing an \hbox{efficient data structure for block sorting}.

\subsection{Mathematical Definition}
\label{subsec:Mathematical Definition}

In this section, we provide a mathematical definition for \texttt{FREIGHT} by expanding the idea of \texttt{Fennel} to the domain of hypergraphs.
Recall that, assuming the vertices of a graph being streamed one-by-one, the \texttt{Fennel} algorithm assigns an incoming vertex~$v$ to a block $V_d$ \hbox{where $d$ is computed as follows}:
\begin{equation}
d = \argmax\limits_{i,~|V_i| < L_{\max}}\big\{|V_i\cap N(v)|-\alpha * \gamma * |V_i|^{\gamma-1}\big\}
\label{eq:Fennel}
\end{equation}

The term $-\alpha * \gamma * |V_i|^{\gamma-1}$, which penalizes block imbalance in \texttt{Fennel}, is directly used in \texttt{FREIGHT} without modification and with the same meaning.
The term $|V_i\cap N(v)|$, which minimizes edge-cut in \texttt{Fennel}, needs to be adapted in~\texttt{FREIGHT} to minimize the intended metric, i.e., either cut-net or connectivity.
Before explaining how this is adapted, recall that, in contrast to graph partitioning, in hypergraph partitioning the incident nets~$I(v)$ of an incoming vertex~$v$ might contain nets that are already cut, i.e., with pins assigned to multiple blocks.
The version of \texttt{FREIGHT} designed to optimize for \emph{connectivity} accounts for already cut nets by keeping track of the block~$d_e$ to which the most recently streamed pin of each net~$e$ has been assigned.
More formally, the connectivity version of \texttt{FREIGHT} assigns an incoming vertex~$v$ of a hypergraph to a block $V_d$ with $d$ given by Equation~(\ref{eq:FREIGHT}), where \hbox{$I^i_{obj}(v) = I^i_{con}(v) =$} \hbox{$\{ e \in I(v): d_e = i\}$}.
On the other hand, the version of \texttt{FREIGHT} designed to optimize for \emph{cut-net} ignores already cut nets, since their contribution to the overall cut-net of the hypergraph $k$-partition is fixed and cannot be changed anymore.
More formally, the cut-net version of \texttt{FREIGHT} assigns an incoming vertex~$v$ of a hypergraph to a block $V_d$ with $d$ given by Equation~(\ref{eq:FREIGHT}), where \hbox{$I^i_{obj}(v) = I^i_{cut}(v) =$} \hbox{$I^i_{con}(v) \setminus E'$}~and~$E'$ is the \hbox{set of already cut nets}.
\begin{equation}
d = \argmax\limits_{i,~|V_i| < L_{\max}}\big\{|I^i_{obj}(v)|-\alpha * \gamma * |V_i|^{\gamma-1}\big\}
\label{eq:FREIGHT}
\end{equation}

Both configurations of \texttt{FREIGHT} interpolate two objectives: favoring blocks with many incident (uncut) nets and penalizing blocks with large cardinality.
We briefly highlight that \texttt{FREIGHT} can be adapted for weighted hypergraphs.
In particular, when dealing with weighted nets, the \hbox{term~$|I^i_{obj}(v)|$} is substituted \hbox{by~$\omega(I^i_{obj}(v))$}.
Likewise when dealing with weighted vertices, the term $-\alpha * \gamma * |V_i|^{\gamma-1}$ is substituted by $-c(v) * \alpha * \gamma * c(V_i)^{\gamma-1}$, where the weight~$c(v)$~of~$v$ is \hbox{used as a multiplicative factor in the penalty term}.

\subsection{Streaming Hypergraphs}
\label{subsec:Models for Streaming Hypergraphs}

In this section, we present and discuss the streaming model used by \texttt{FREIGHT}.
Recall in the streaming model for graphs vertices  are loaded one at a time alongside with their adjacency lists.
Thus, just streaming the graph (without doing additional compuations, implies a time cost $O(m+n)$.
In our model, the vertices of a hypergraph are loaded one at a time alongside with their incident nets, as illustrated in Figure~\ref{fig:StreamingModel}.
Our streaming model implies a cost $O(\sum_{e \in E}{|e|} + n)$ to stream the hypergraph, where $O(\sum_{e \in E}{|e|})$ is the cost to stream each net $e$ exactly $|e|$ times.
\texttt{FREIGHT} uses $O(m+k)$ memory, with $O(m)$ being used to keep track, for each net~$e$, of its cut/uncut status as well as the block~$d_e$ to which its most recently streamed pin was assigned.
This net-tracking information, which substitutes the need to keep track of vertex assignments, is necessary for executing \texttt{FREIGHT}. 
Although \texttt{FREIGHT} consumes more memory than required by graph-based streaming algorithms which often use $O(m+k)$ memory, it is still far better than the $O(mk)$ worst-case memory required by the state-of-the-art algorithms for streaming hypergraph partitioning~\cite{alistarh2015streaming,tacsyaran2021streaming}, all of which are also based on a computational model that implies a cost \hbox{$O(\sum_{e \in E}{|e|} + n)$} \hbox{just to stream the hypergraph}.

\subsection{Efficient Implementation}
\label{subsec:Efficient Implementation}

In this section, we describe an efficient implementation for \texttt{FREIGHT}.
Recall that, for every vertex $v$ that is loaded, \texttt{FREIGHT} uses Equation~(\ref{eq:FREIGHT}) to find the block with the highest score among \NEW{up to}~$k$~options.
A simple method to accomplish this task consists of explicitly evaluating the score for each block and identifying the one with the highest score.
This results in a total of $O(nk)$ evaluations, leading to an overall complexity of \hbox{$O(\sum_{e \in E}{|e|}+nk)$}.
We propose an implementation that is significantly more efficient \hbox{than this approach}.

For each loaded vertex~$v$, our implementation separates the blocks \NEW{$V_i$ for which \hbox{$|V_i|<L_{\max}$}} into two disjoint sets,~$S_1$~and~$S_2$.
In particular, the set~$S_1$ comprises blocks $V_i$ where \hbox{$|I^i_{obj}(v)|>0$}, while the set~$S_2$ comprises the remaining blocks, i.e., blocks $V_i$ for which \hbox{$|I^i_{obj}(v)|=0$}.
Using the sets provided, we break down Equation~(\ref{eq:FREIGHT}) into Equation~(\ref{eq:FREIGHT_E_1}) and Equation~(\ref{eq:FREIGHT_E_2}), which are solved separately. 
The resulting solutions are compared based on their \texttt{FREIGHT} scores to ultimately find the solution for Equation~(\ref{eq:FREIGHT}).
The overall process is \hbox{illustrated in~Figure~\ref{fig:DecomposeEquation}}.
\begin{equation}
d = \argmax\limits_{i \in S_1}\big\{|I^i_{obj}(v)|-\alpha * \gamma * |V_i|^{\gamma-1}\big\}
\label{eq:FREIGHT_E_1}
\end{equation}
\begin{equation}
d = \argmax\limits_{i \in S_2}\big\{|I^i_{obj}(v)|-\alpha * \gamma * |V_i|^{\gamma-1}\big\} = \argmin\limits_{i \in S_2}|V_i|
\label{eq:FREIGHT_E_2}
\end{equation}

\begin{figure}[t]
	\centering
	\includegraphics[width=0.9\linewidth]{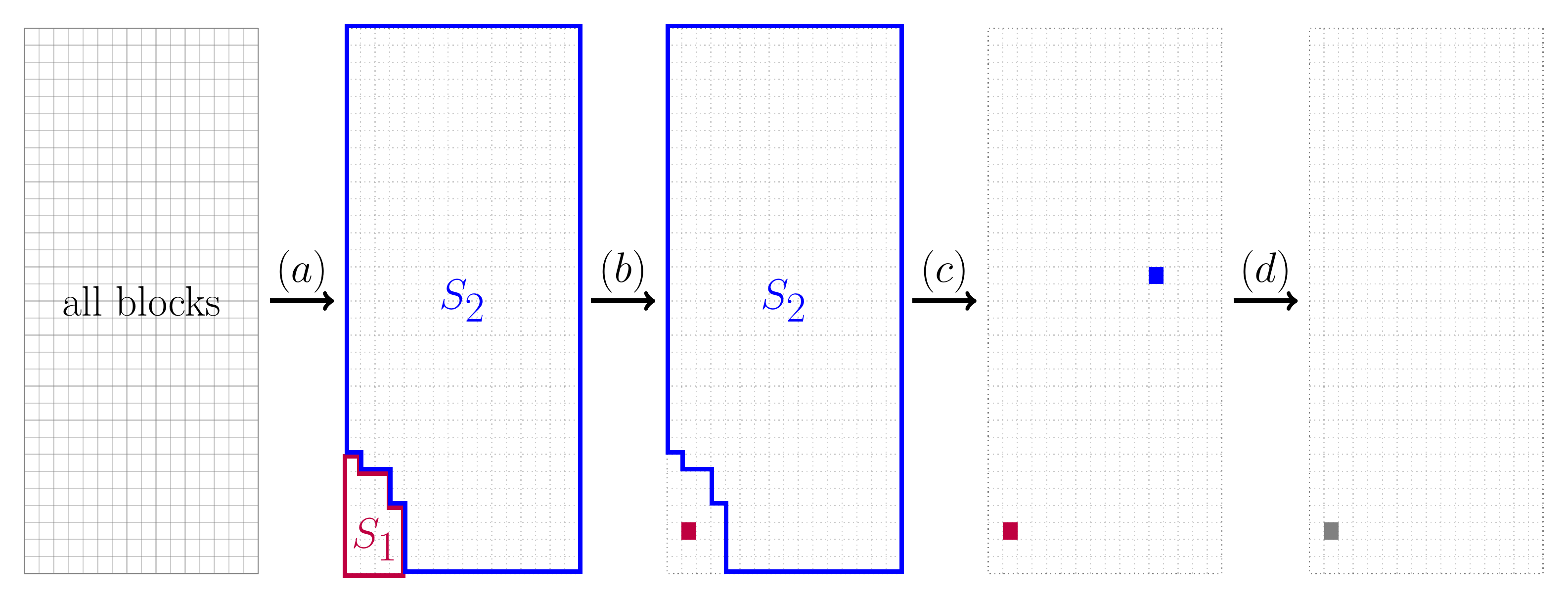}
	\caption{Illustration of the process to solve Equation~(\ref{eq:FREIGHT}) for an incoming vertex~$u$ with \hbox{$k=512$}~blocks. 
		(a)~The~$k$~blocks are decomposed into~$S_1$~and~$S_2$, with~\hbox{$|S_1| = O(|I(u)|)$}. 
		(b)~Equation~(\ref{eq:FREIGHT_E_1}) is explicitly solved at cost~\hbox{$O(|I(u)|)$}.  
		(c)~Equation~(\ref{eq:FREIGHT_E_2}) is implicitly solved at cost~$O(1)$.
		(d)~Both solutions are then evaluated using their \texttt{FREIGHT} scores to determine \hbox{the final solution for Equation~(\ref{eq:FREIGHT})}.}
	\label{fig:DecomposeEquation}
\end{figure}

Now we explain how we solve Equation~(\ref{eq:FREIGHT_E_1})~and~Equation~(\ref{eq:FREIGHT_E_2}).
To solve Equation~(\ref{eq:FREIGHT_E_1}), we use the theoretical complexity outlined in Theorem~\ref{theo:E_1} and solve it explicitly. 
In contrast, Equation~(\ref{eq:FREIGHT_E_2}) is implicitly solved by identifying the block with minimal cardinality.
We use an efficient data structure to keep all blocks sorted by cardinality throughout the entire execution, which enables us to solve \hbox{Equation~(\ref{eq:FREIGHT_E_2}) in constant time}.

\begin{theorem}
	Equation~(\ref{eq:FREIGHT_E_1}) can be solved in time~$O(|I(v)|)$.
	\label{theo:E_1}
\end{theorem}

\begin{proof}
	The terms \hbox{$|I^i_{obj}(v)|$} in Equation~(\ref{eq:FREIGHT_E_1}) can be computed by iterating through the nets of~$v$ at a cost of \hbox{$O(|I(v)|)$} and determining their status as cut, unassigned, or assigned to a block.
	The calculation of the factors \hbox{$-\alpha * \gamma * |V_i|^{\gamma-1}$} in Equation~(\ref{eq:FREIGHT_E_1}) can be done in time \hbox{$O(|S_1|) = O(|I(v)|)$}, \hbox{thus completing the proof}.
\end{proof}

Now we explain our data structure to keep the blocks sorted by cardinality during the whole algorithm execution.
The data structure is implemented with two arrays~$A$~and~$B$, both with $k$ elements, and a list~$L$.
The array~$A$ stores all $k$~blocks always in ascending order. 
The array~$B$ maps the index~$i$ of a block~$V_i$ to its position in~$A$.
Each element in the list~$L$ represents a bucket.
Each bucket is associated with a unique block cardinality and contains the leftmost and the rightmost positions~$\ell$~and~$r$ of the range of blocks in~$A$ which currently have this cardinality.
Reciprocally, each block in~$A$ has a pointer to the unique bucket in~$L$ corresponding to its cardinality.
To begin the algorithm, $L$ is set up with a single bucket for cardinality~$0$ which covers the $k$ positions of~$A$, i.e., its paramenters $\ell$~and~$r$ are $1$~and~$k$, respectively. 
The blocks in~$A$ are sorted in any order initially, however, as each block starts with a cardinality of $0$, they will be \hbox{ordered by their cardinalities}.

\begin{figure}[t]
	\centering
	\includegraphics[width=0.7\linewidth]{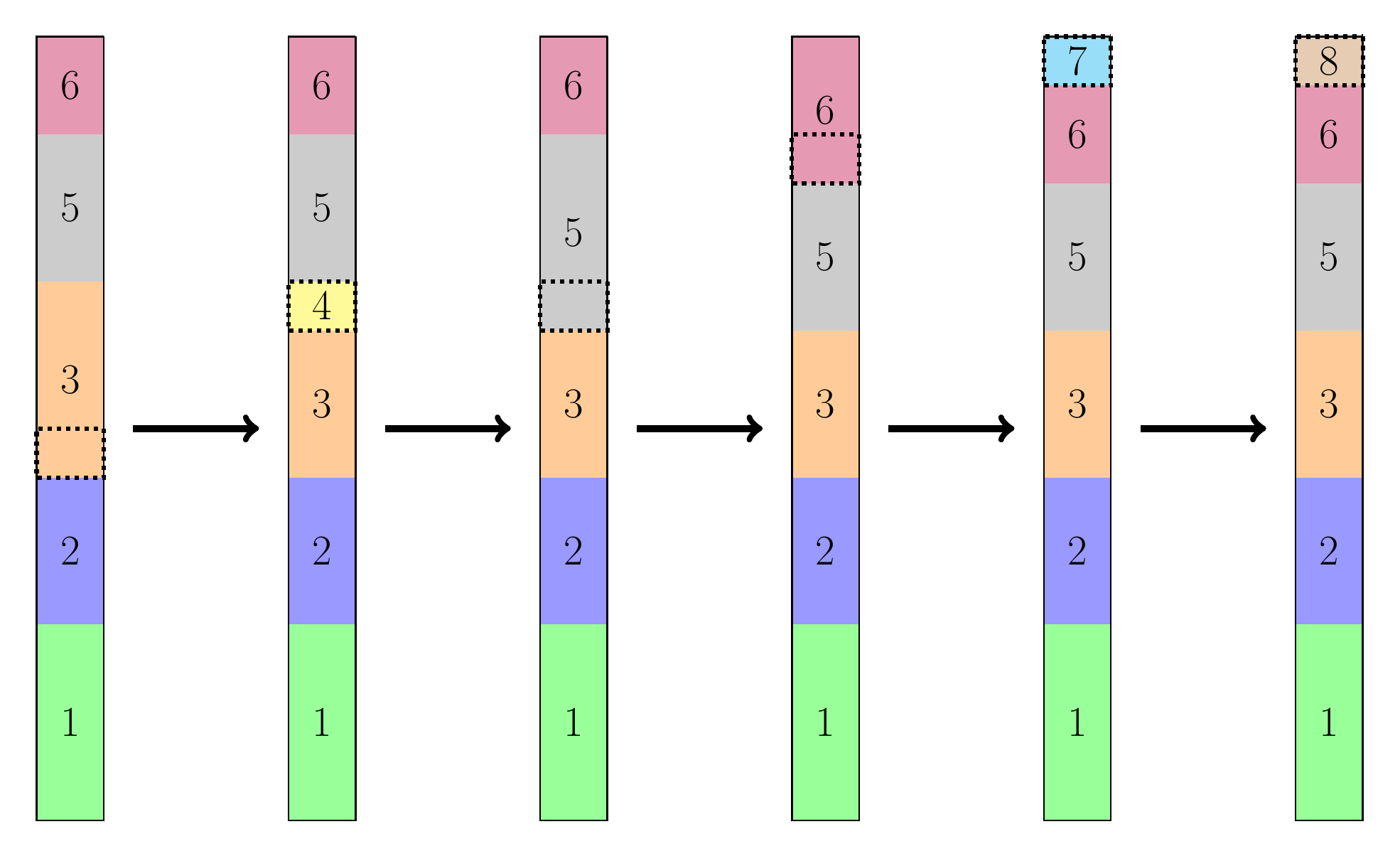}
	\caption{Illustration of our data structure used to keep the blocks sorted by cardinality throughout the execution of \texttt{FREIGHT}. The array~$A$ is represented as a vertical rectangle. Each region of~$A$ is covered by a unique bucket, which is represented by a unique color filling the corresponding region in~$A$. The cardinality associated with each bucket is written in the middle of the region of~$A$ covered by it. Here we represent the behavior of the data structure when assigning vertices \hbox{to the block surrounded by a dotted rectangle five times consecutively}.}
	\label{fig:SortingArray}
\end{figure}

When a vertex is assigned to a block~$V_d$, we update our data structure as detailed in Algorithm~\ref{alg:IncrementCardinalityBlock} and exemplified in Figure~\ref{fig:SortingArray}.
We describe Algorithm~\ref{alg:IncrementCardinalityBlock} in detail now.
In line~1, we find the position~$p$ of~$V_d$ in~$A$ and find the bucket~$C$ associated with it.
In line~2, we exchange the content of two positions in~$A$: the position where~$V_d$ is located and the position identified by the variable $r$~in~$C$, which marks the rightmost block in~$A$ covered by~$C$.
This variable $r$ is afterwards decremented in line~3 since $V_d$ is now not covered anymore by the bucket $C$.
In lines~4~and~5, we check if the new (increased) cardinality of~$V_d$ matches the cardinality of the block located right after it in~$A$.
If so, we associate $V_d$ to the same bucket as it and decrement this bucket's leftmost position~$\ell$ in line~6;
Otherwise, we push a new bucket to~$L$ and match it to $V_d$ adequately in lines~8~and~9.
Finally, in line~10, we delete~$C$ in case its range $[\ell,r]$ is empty.
Figure~\ref{fig:SortingArray} shows our data structure through five consecutive executions of Algorithm~\ref{alg:IncrementCardinalityBlock}.
Theorem~\ref{theo:sorting_correctness} proves the correctness of our data structure. 
Theorem~\ref{theo:E_2} shows that, using our proposed data structure, we need time~$O(1)$ to either solve Equation~(\ref{eq:FREIGHT_E_2}) or prove that the solution for Equation~(\ref{eq:FREIGHT_E_1}) solves Equation~(\ref{eq:FREIGHT}).
\NEW{It should be noted that using a bucket queue would result in the same complexity, but it would require $O(k+L_{\max})$ of memory, whereas our data structure only uses $O(k)$.}
The overall complexity of \texttt{FREIGHT}, which directly follows from Theorem~\ref{theo:E_1} and Theorem~\ref{theo:E_2}, \hbox{is expressed in Corollary~\ref{cor:complexity_without_sampling}}.

\begin{algorithm}[] %
	\caption{Increment cardinality of block $V_d$ in the proposed data structure} %
	\label{alg:IncrementCardinalityBlock} %
	\begin{algorithmic}[1] %
		\State $p \gets B_d$; $C \gets A_p.bucket$;
		\State $q \gets C.r$; $c \gets A_q.id$; $Swap(A_p, A_q)$; $Swap(B_c, B_d)$; 
		\State $C.r \gets C.r - 1$;
		\State $C^\prime \gets A_{q+1}.bucket$;
		\If{$C.cardinality + 1 = C^\prime.cardinality$}
		\State $A_q.bucket \gets C^\prime$; $C^\prime.\ell \gets C^\prime.\ell - 1$;
		\Else
		\State $C^{\prime\prime} \gets NewBucket()$; $A_q.bucket \gets C^{\prime\prime}$; $L \gets L \cup \{C^{\prime\prime}\}$; 
		\State $C^{\prime\prime}.cardinality \gets C.cardinality + 1$; $C^{\prime\prime}.\ell \gets q$; $C^{\prime\prime}.r \gets q$;
		\EndIf
		\State \textbf{if} $C.r = C.\ell$ \textbf{then} $L \gets L \setminus \{C\}$;
	\end{algorithmic}
\end{algorithm}

\begin{theorem}
	Our proposed data structure keeps the blocks within array A consistently sorted in ascending order of cardinality.
	\label{theo:sorting_correctness}
\end{theorem}

\begin{proof}
	We inductively prove two claims at the same time: 
	(a)~the variables~$\ell$~and~$r$ contained in each bucket from~$L$ respectively store the leftmost and the rightmost positions of the unique range of blocks in~$A$ which currently have this cardinality;
	(b)~the array~$A$ contains the blocks sorted in ascending order of cardinality.
	Both claims are trivially true at the beginning, since all blocks have cardinality 0 and~$L$ is initialized with a single bucket with $\ell=1$ and $r=k$.
	Now assuming that (a)~and~(b) are true at some point, we show that they keep being true after Algorithm~\ref{alg:IncrementCardinalityBlock} is executed.
	Note that line~2 performs the only position exchange in~$A$ throughout the whole algorithm.
	As~(a) is assumed, it is the case that~$V_d$ swaps positions with the rightmost block in $A$ containing the same cardinality of~$V_d$.
	Since the cardinality of~$V_d$ will be incremented by one and all blocks have integer cardinalities, this concludes the proof of~(b).
	To prove that~(a) remains true, note that the only buckets in~$L$ that are modified are~$C$ (line~3),~$C^\prime$ (line~6), and~$C^{\prime\prime}$ (line~9).
	Claim~(a) remains true for~$C$ because $V_d$, whose cardinality will be incremented, is the only block removed from its range.
	Claim~(a) remains true for~$C^\prime$ because line~6 is only executed if the new cardinality of~$V_d$ equals the cardinality of~$C^\prime$, whose current range starts right after the new position of~$V_d$ in~$A$.
	Bucket~$C^{\prime\prime}$ is only created if the new cardinality of~$V_d$ is respectively larger and smaller than the cardinalities of~$C$~and~$C^\prime$.
	Since~(b) is true, then this condition only happens if there is no block in~$A$ with the same cardinality as the new cardinality of~$V_d$.
	Hence, claim~(a) remains true for~$C^{\prime\prime}$, which is created \hbox{covering only the position of~$V_d$ in~$A$}.
\end{proof}

\begin{theorem}
	By utilizing our proposed data structure, solving Equation~(\ref{eq:FREIGHT_E_2}) or demonstrating that any solution for Equation~(\ref{eq:FREIGHT_E_1}) is also a solution for Equation~(\ref{eq:FREIGHT}) can be accomplished~in~$O(1)$~time.
	\label{theo:E_2}
\end{theorem}

\begin{proof}
	Algorithm~\ref{alg:IncrementCardinalityBlock} contains no loops and each command in it has a complexity of $O(1)$, thus the total cost of the algorithm is $O(1)$.
	Our data structure executes Algorithm~\ref{alg:IncrementCardinalityBlock} once for each assigned vertex, hence it costs $O(1)$ per vertex.
	Say we are evaluating an incoming vertex~$v$.
	According to Theorem~\ref{theo:sorting_correctness}, the block~$V_d$ with minimum cardinality is stored in the first position of the array~$A$, hence it can be accessed in time $O(1)$.
	In~case~\hbox{$V_d \in S_2$}, then~$d$ is a solution for Equation~(\ref{eq:FREIGHT_E_2}).
	On the other hand, if $V_d$ is in $S_1$, the \texttt{FREIGHT} score of $V_d$ will be larger than the \texttt{FREIGHT} score of the solution for Equation~(\ref{eq:FREIGHT_E_2}) by at least \hbox{$|I_d(v)| > 0$}.
	In this case, it follows that \hbox{any solution for Equation~(\ref{eq:FREIGHT_E_1}) solves Equation~(\ref{eq:FREIGHT})}.
\end{proof}

\begin{corollary}
	The overall complexity of \texttt{FREIGHT} is \hbox{$O\big(\sum_{e \in E}{|e|} + n\big)$}.
	\label{cor:complexity_without_sampling}
\end{corollary}

\newcommand{ \scaleFactorSmall} {0.49}
\newcommand{ \imgScaleFactorSmall} {1}
\newcommand{ \capPositionSmall} {-.4cm}
\newcommand{ \afterCapSmall} {-.45cm}

\section{Experimental Evaluation}
\label{sec:Experimental Evaluation}

\vspace*{-0.25cm}
\subparagraph*{Setup.} 
We performed our implementations in C++ and compiled them using gcc 11.2 with full optimization turned on (-O3 flag). 
Unless mentioned otherwise, all experiments are performed on a single core of a machine consisting of a sixteen-core Intel Xeon  Silver 4216 processor running at $2.1$ GHz, $100$ GB of main memory, $16$ MB of L2-Cache, and $22$ MB of L3-Cache running Ubuntu 20.04.1. 
The machine can handle 32 threads with hyperthreading. %
Unless otherwise mentioned we stream (hyper)graphs directly from the internal memory to obtain clear running time comparisons. 
However, note that \texttt{FREIGHT} as well as most of the other used algorithms can also be run streaming \hbox{the hypergraphs from hard~disk}.

\vspace*{-.25cm}
\subparagraph*{Baselines.}
We compare \texttt{FREIGHT} against various state-of-the-art algorithms.
In this section we will list these algorithms and explain our criteria for algorithm selection.
We have implemented \texttt{Hashing}  in C++, since it is a simple algorithm.
\NEW{It basically consists of hashing the IDs of incoming vertices into $\{1,\ldots,k\}$.}
The remaining algorithms were obtained either from official repositories or privately from the authors, with the exception of  \texttt{Min-Max}, for which there is no official implementation available.
Here, we use the \texttt{Min-Max} implementations by Ta{\c s}yaran~et~al.~\cite{tacsyaran2021streaming}. 
\hbox{All algorithms were compiled with gcc 11.2}.

We run \texttt{Hashing}, \texttt{Min-Max}~\cite{alistarh2015streaming} and all its improved versions proposed by Ta{\c s}yaran~et~al.~\cite{tacsyaran2021streaming}:
\texttt{Min-Max-BF}, \texttt{Min-Max-N2P}, \texttt{Min-Max-L$\ell$}, \texttt{Min-Max-MH}, \texttt{REF}, \texttt{REF\_RLX}, and \texttt{REF\_RLX\_SV}.
(see Section~\ref{subsec:related_work} for details on the different \texttt{Min-Max} versions), \texttt{HYPE}~\cite{HYPE2018}, and \texttt{PaToH}  v3.3~\cite{ccatalyurek2011patoh}.
\texttt{Hashing}  is relevant because it is the simplest and fastest streaming algorithm, which gives us a lower bound for partitioning time.
\texttt{Min-Max} is a current state-of-the-art for streaming hypergraph partitioning in terms of cut-net and connectivity.
The improved and buffered versions of \texttt{Min-Max} proposed in~\cite{tacsyaran2021streaming} are relevant because some of them are orders of magnitude faster than \texttt{Min-Max} while others produce improved partitions in comparison to it.
\texttt{HYPE}   and \texttt{PaToH}  are in-memory algorithms for hypergraph partitioning, hence they are not suitable for the streaming setting.
However, we compare against them because \texttt{HYPE}   is among the fastest in-memory algorithms while \texttt{PaToH}  is very fast and also computes partitions with very good cut-net and connectivity.
Note that KaHyPar~\cite{schlag2016k} is the leading tool with respect to solution quality, however it is also much slower than \texttt{PaToH}.

\vspace*{-.25cm}
\subparagraph*{Instances.}
We selected hypergraphs from various sources to test our algorithm.
The considered hypergraphs were used for benchmark in previous works on hypergraph partitioning.
Prior to each experiment, we converted all hypergraphs to the appropriate streaming formats required by each algorithm.
We removed parallel and empty hyperedges and self loops, and assigned unitary weight to all vertices and hyperedges.
In all experiments with streaming algorithms, we stream the hypergraphs with the natural given order of the vertices.
We use a number of blocks \NEW{$k \in \{512,1024,1536,2048,2560\}$} 
unless mentioned otherwise.
We allow a fixed imbalance of $3\%$ for all experiments (and all algorithms) since this is a frequently used value in the partitioning literature. 
All algorithms always generated balanced partitions, except for \texttt{HYPE}  
which generated highly \hbox{unbalanced partitions in around $5\%$ of its experiments}.

We use the same benchmark as in~\cite{schlag2016k}.
This consists of 310~hypergraphs from three benchmark sets: 
18 hypergraphs from the ISPD98 Circuit Benchmark Suite~\cite{alpert1998ispd98}, 
192 hypergraphs based on the University of Florida Sparse Matrix Collection~\cite{davis2011university}, and 
100 instances from the international SAT Competition 2014~\cite{Belov2014}. 
The SAT instances were converted into hypergraphs by mapping each boolean variable and its complement to a vertex and each clause to a net. 
From the Sparse Matrix Collection, one matrix was selected for each application area that had between \numprint{10000} and \numprint{10000000} columns. 
The matrices were converted into hypergraphs using the row-net model, in which each row is \hbox{treated as a net and each column as a vertex}.

\vspace*{-0.25cm}
\subparagraph*{Methodology.}
Depending on the focus of the experiment, we measure running time, cut-net, and-or connectivity.
We perform 5 repetitions per algorithm and instance using random seeds for non-deterministic algorithms, and calculate the arithmetic average of the computed objective function and running time per instance.
When further averaging over multiple instances, we use the geometric mean in order to give every instance \hbox{the same influence on the \textit{final score}}. 

Given a result of an algorithm~$A$, we express its value $\sigma_A$ (which can be objective or running time) as \emph{improvement} over an algorithm~$B$, computed as $\big(\frac{\sigma_B}{\sigma_A}-1\big)*100\%$;
We also use \emph{performance profiles} to represent results.
They relate the running time (quality) of a group of algorithms to the fastest (best) one on a per-instance basis (rather than grouped by $k$).
The x-axis shows a factor $\tau$ while the y-axis shows the percentage of instances for which A has up to $\tau$ times the running time (quality) of the fastest (best)~algorithm.
Bar charts and boxplots are also employed to represent our findings.
We use bar charts to visualize the average value of an objective function in relation to $k$, where each algorithm is represented by vertical bars of a given color with origin on the x-axis. 
The bars for every value of $k$ have a common origin and are arranged in terms of their height, allowing all heights to be visible. 
We use boxplots to give a clear picture of the dataset distribution by displaying the minimum, maximum, median, first and third quartiles, \hbox{while disregarding outliers}.

\begin{figure*}[p!]
	\captionsetup[subfigure]{justification=centering}
	\centering
	\begin{subfigure}[]{\scaleFactorSmall\textwidth}
		\centering
		\includegraphics[width=\imgScaleFactorSmall\textwidth]{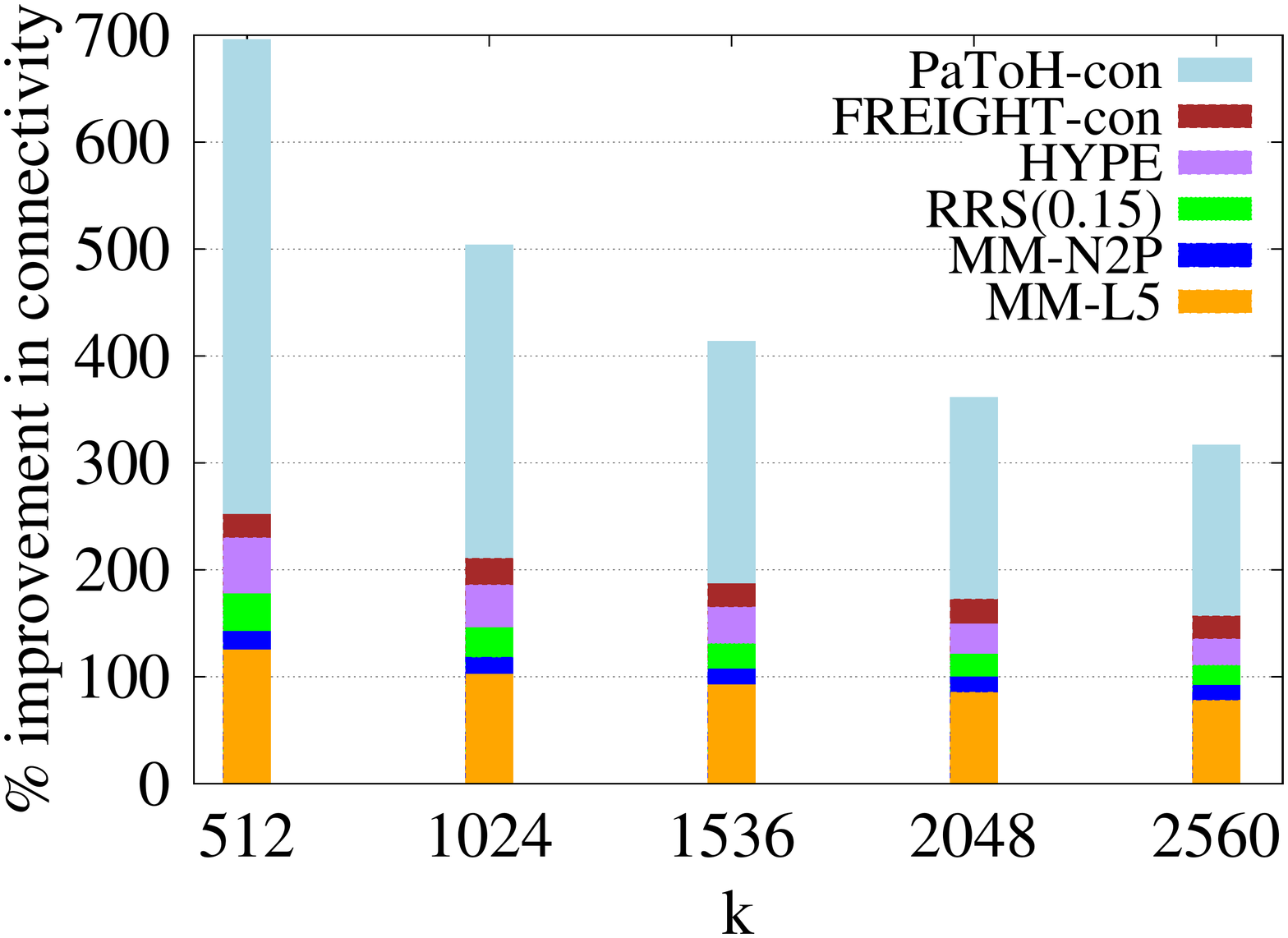}
		\vspace*{\capPositionSmall}
		\caption{Connectivity ioH.}
		\label{fig:connectivity_impr}
	\end{subfigure}%
	\begin{subfigure}[]{\scaleFactorSmall\textwidth}
		\centering
		\includegraphics[width=\imgScaleFactorSmall\textwidth]{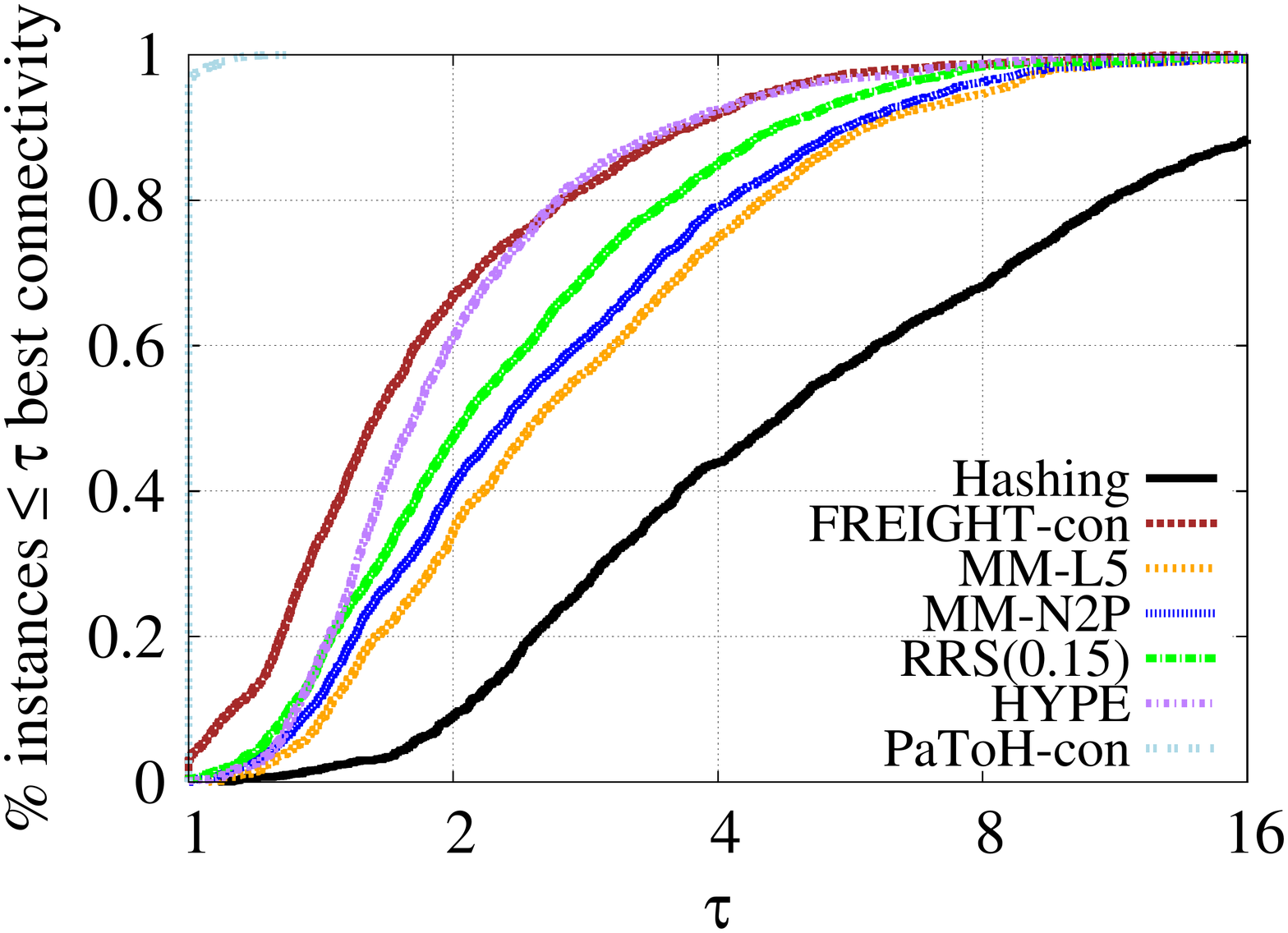}
		\vspace*{\capPositionSmall}
		\caption{Connectivity pp.}
		\label{fig:connectivity_pp}
	\end{subfigure}%

	\begin{subfigure}[]{\scaleFactorSmall\textwidth}
		\centering
		\includegraphics[width=\imgScaleFactorSmall\textwidth]{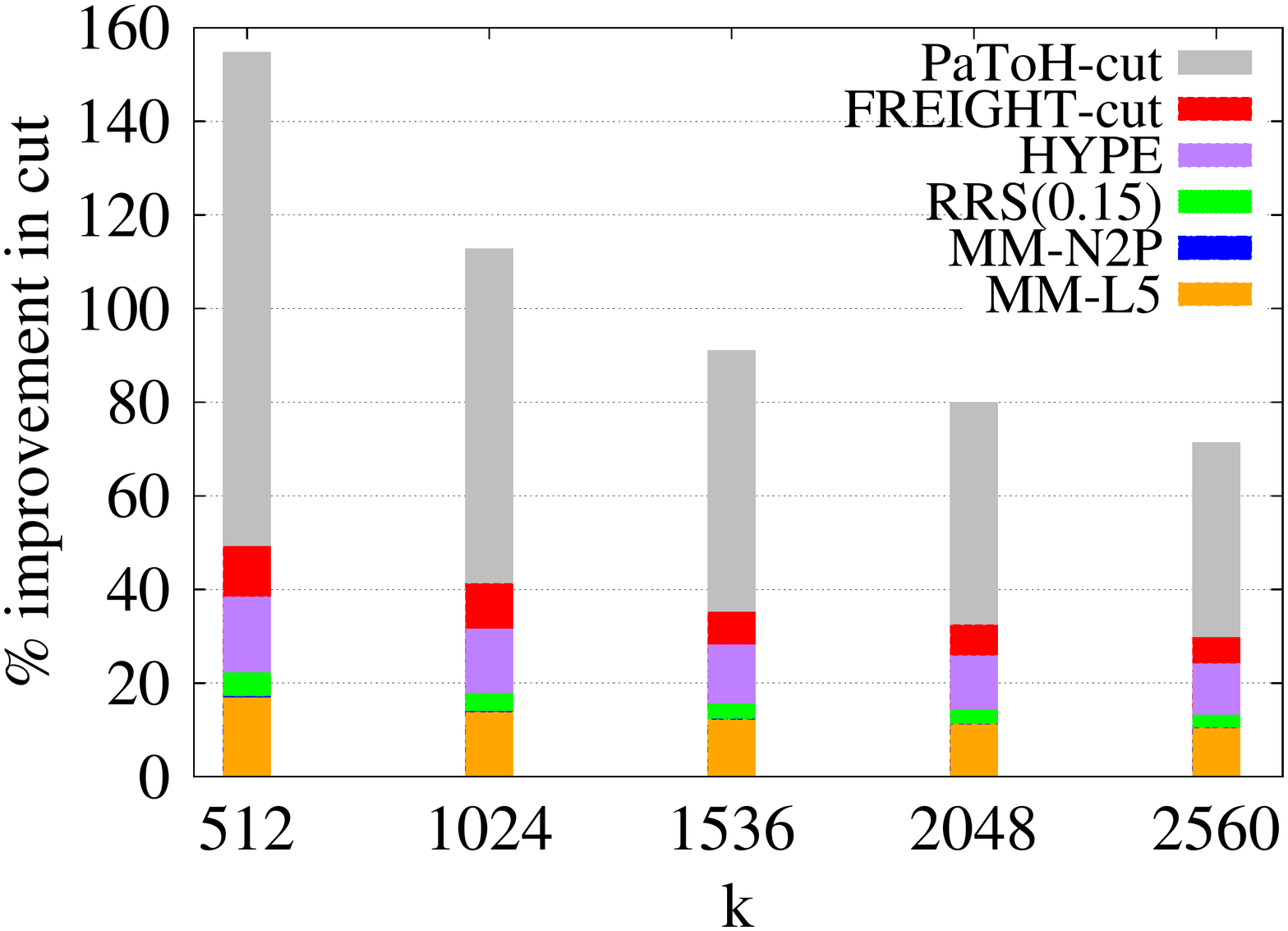}
		\vspace*{\capPositionSmall}
		\caption{Cut-net ioH.}
		\label{fig:cut_impr}
	\end{subfigure}%
	\begin{subfigure}[]{\scaleFactorSmall\textwidth}
		\centering
		\includegraphics[width=\imgScaleFactorSmall\textwidth]{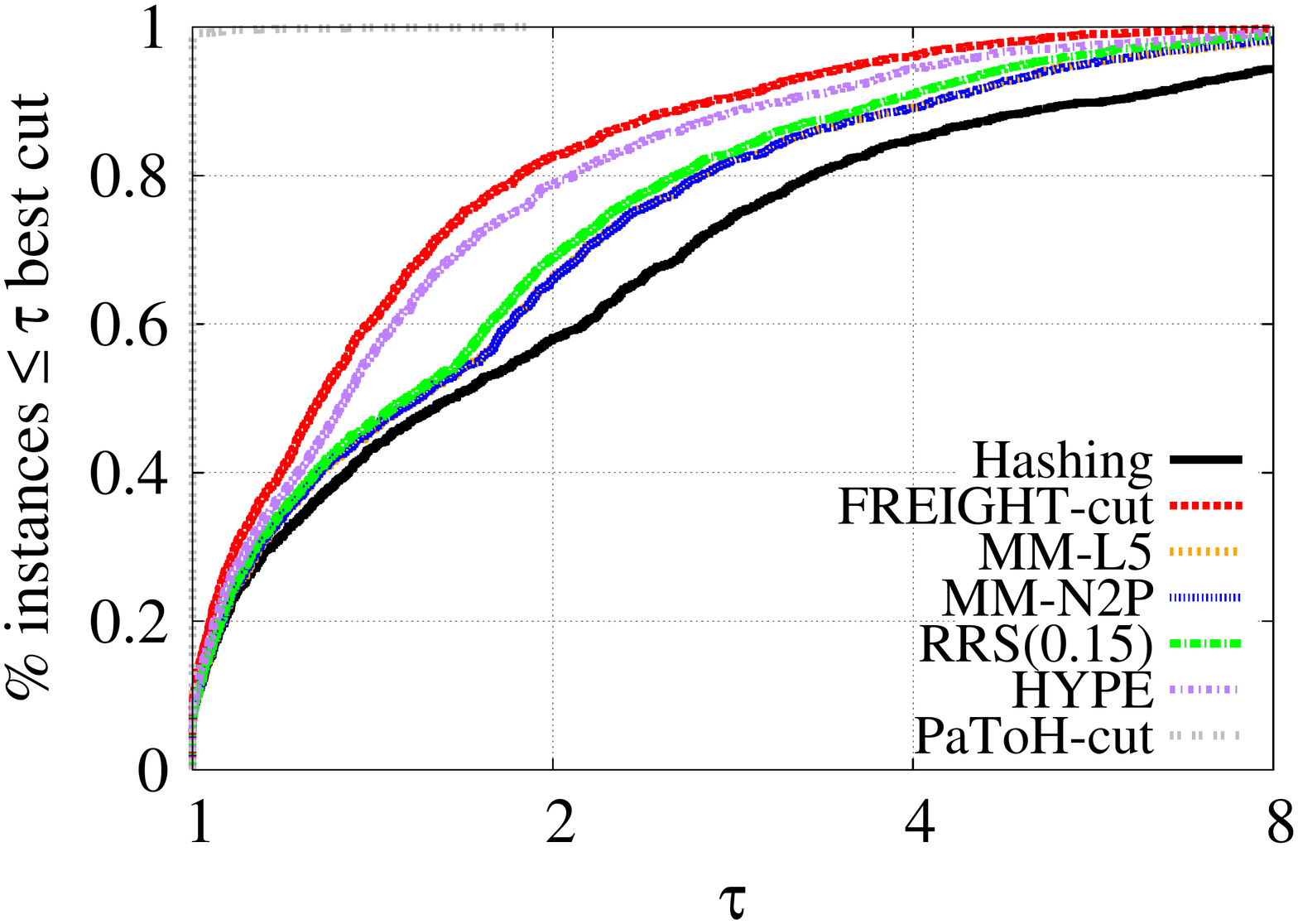}
		\vspace*{\capPositionSmall}
		\caption{Cut-net pp.}
		\label{fig:cut_pp}
	\end{subfigure}

	\begin{subfigure}[]{\scaleFactorSmall\textwidth}
		\centering
		\includegraphics[width=\imgScaleFactorSmall\textwidth]{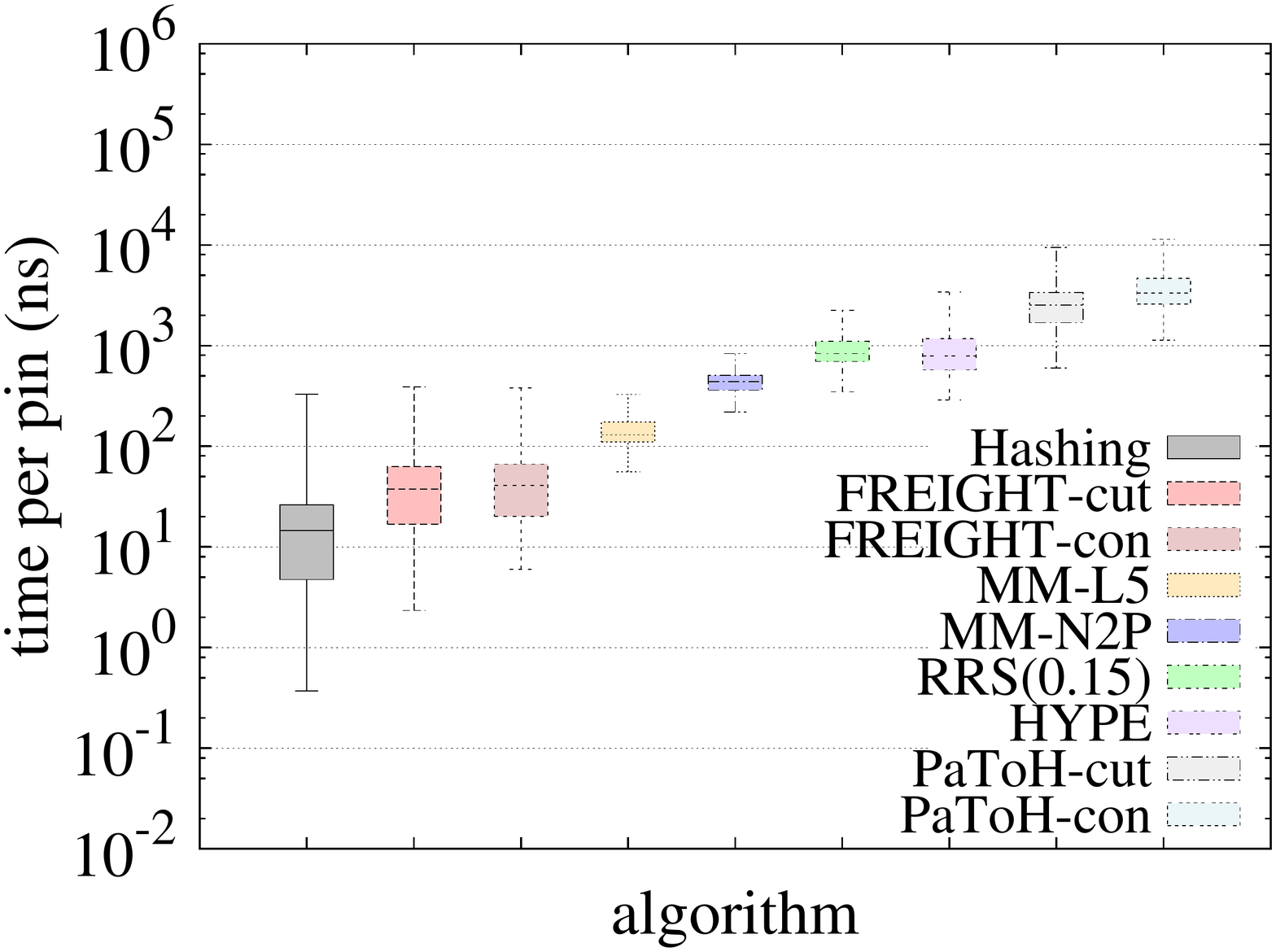}
		\vspace*{\capPositionSmall}
		\caption{Running time bp.}
		\label{fig:time_bp}
	\end{subfigure}
	\begin{subfigure}[]{\scaleFactorSmall\textwidth}
		\centering
		\includegraphics[width=\imgScaleFactorSmall\textwidth]{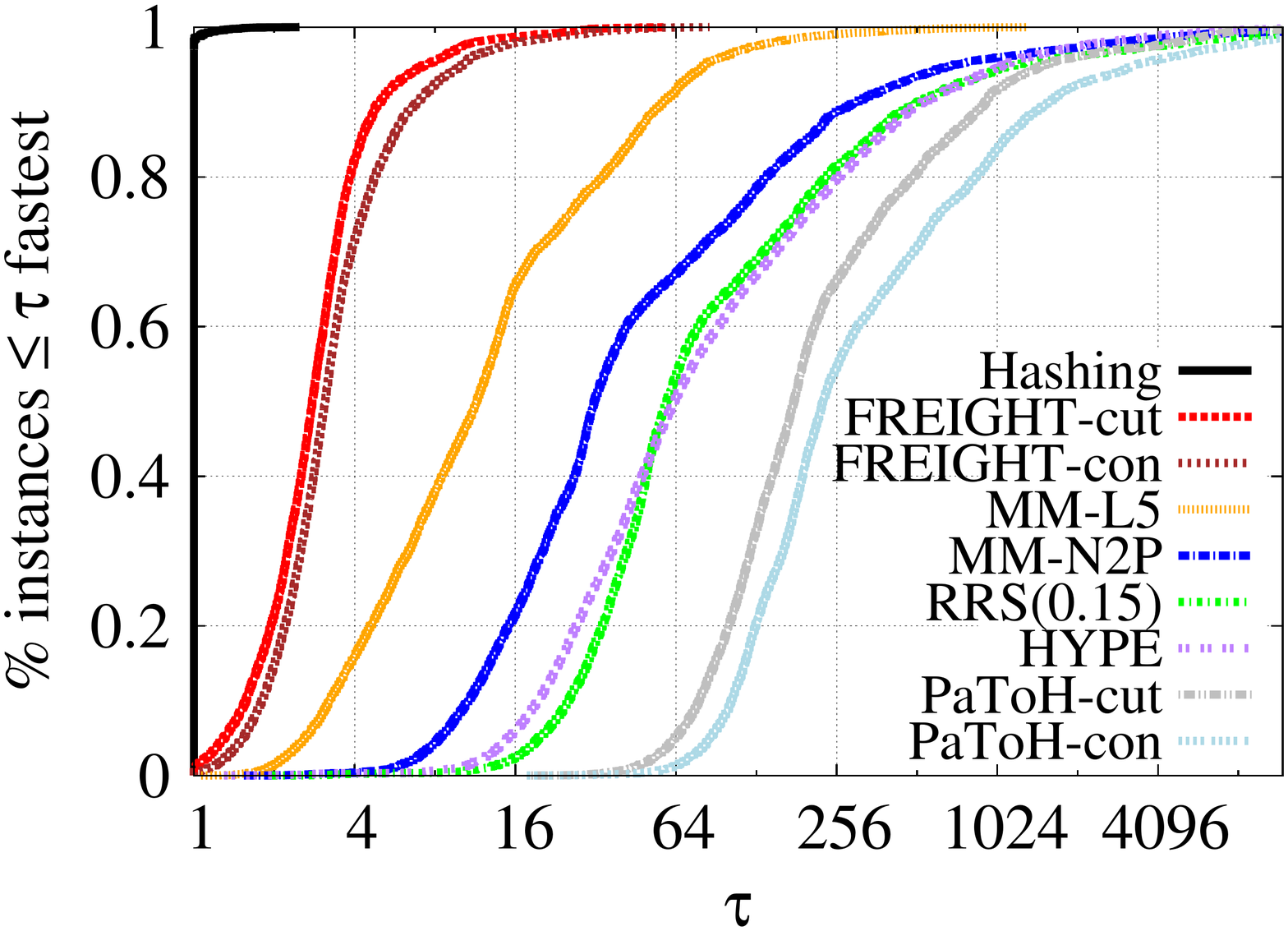}
		\vspace*{\capPositionSmall}
		\caption{Running time pp.}
		\label{fig:time_pp}
	\end{subfigure}%
	
	\vspace*{0.5cm}
	
	\caption{Comparison against the state-of-the-art streaming algorithms for hypergraph partitioning. We show performance profiles~(pp), improvement plots over \texttt{Hashing}~(ioH), and boxplots~(bp). Note that \texttt{PaToH-con}, \texttt{PaToH-cut}, and \texttt{Hashing}  align almost perfectly with the y-axis in Figures~\ref{fig:connectivity_pp},~\ref{fig:cut_pp}, and~\ref{fig:time_pp}, respectively. Also the curves and bars of \texttt{MM-N2P} and \texttt{MM-L5} roughly overlap \hbox{with one another in Figure~\ref{fig:cut_pp} and Figure~\ref{fig:cut_impr}}.}
	\label{fig:state-of_the_art}
	
\end{figure*}

 \vspace*{-.25cm}
\subsection{Results}
\label{subsec:Results}

In this section, we show experiments in which we compare \texttt{FREIGHT} against the current state-of-the-art of streaming hypergraph partitioning.	
As already mentioned, we also use two internal-memory algorithms~\cite{HYPE2018,ccatalyurek2011patoh} as more general baselines for comparison. 
We focus our experimental evaluation on the comparison of solution quality and running time.
Observe that \texttt{PaToH}  and \texttt{FREIGHT} have distinct versions designed to optimize for each quality metric (i.e., connectivity and cut-net).
For a meaningful comparison, we only take into account the relevant version when dealing with each quality metric, however, both versions are still considered for running time comparisons.
To differentiate between the versions, suffixes \texttt{-con} and \texttt{-cut} are added to represent the connectivity-optimized and cut-net versions respectively.
For clarity, we refrain from discussing state-of-the-art streaming algorithms that are \emph{dominated} by another algorithm.
We define a dominated algorithm as one that has worse running time compared to another without offering a superior solution quality in return, or vice-versa.
In particular, we leave out \texttt{Min-Max} and \texttt{Min-Max-BF} since they are dominated by \texttt{Min-Max-N2P}, which is referred to as \texttt{MM-N2P} hereafter. %
Similarly, we omit \texttt{Min-Max-MH} 
because \NEW{it is} dominated by \texttt{Hashing}.
We use a buffer size of $15\%$ for testing the buffered algorithms \texttt{REF}, \texttt{REF\_RLX}, and \texttt{REF\_RLX\_SV}, following the best results outlined in~\cite{tacsyaran2021streaming}. %
We omit the first two of them since they are dominated by the latter one, which is referred to as \texttt{RRS(0.15)}   from now on.
Since \texttt{Min-Max-L$\ell$} is not dominated by any other algorithm, we exhibit its results with $\ell=5$, as seen in the best results in~\cite{tacsyaran2021streaming}, and \hbox{we refer to it as \texttt{MM-L5} from this point}.

\vspace*{-0.25cm}
\subparagraph*{Connectivity.}
We start by looking at the connectivity metric. 
In Figure~\ref{fig:connectivity_impr}, we plot the average connectivity improvement over \texttt{Hashing}  for each value of~$k$. 
\texttt{PaToH-con}  produces the best connectivity on average, yielding an average improvement of $443\%$ when compared to \texttt{Hashing}.
This is in line with previous works in the area of (hyper)graph partitioning, i.e. streaming algorithms typically compute worse solutions than internal memory algorithms, which have access to the whole graph. 
\texttt{FREIGHT-con} is found to be the second best algorithm in terms of connectivity, outperforming both the internal memory algorithm \texttt{HYPE}   and the buffered streaming algorithm \texttt{RRS(0.15)}.
On average, these three algorithms improve $194\%$, $171\%$, and $136\%$ over \texttt{Hashing}, respectively.
Finally, \texttt{MM-N2P} and \texttt{MM-L5} compute solutions which improve $111\%$ and $96\%$ over \texttt{Hashing}  on average, respectively.
In direct comparison, \texttt{FREIGHT-con} shows average connectivity improvements of $8\%$, $24\%$, $39\%$, and $50\%$ over \texttt{HYPE}, \texttt{RRS(0.15)}, \texttt{MM-N2P}, and \texttt{MM-L5}, respectively.
Note that each algorithm \hbox{retains its relative ranking in terms of average connectivity over all values of~$k$}.

In Figure~\ref{fig:connectivity_pp}, we plot connectivity performance profiles across all experiments.
\texttt{PaToH-con}  produces the best overall connectivity for $96.4\%$ of the instances, while \texttt{FREIGHT-con} produces the best connectivity for $3.1\%$ of the instances and no other algorithm computes the best connectivity for more than $0.35\%$ of the instances.
The connectivity produced by  \texttt{FREIGHT-con}, \texttt{HYPE}, \texttt{RRS(0.15)}, \texttt{MM-N2P}, \texttt{MM-L5}, and \texttt{Hashing}  are within a factor~$2$ of the best found connectivity for $67\%$, $61\%$, $47\%$, $41\%$, $34\%$, and $9\%$ of the instances, respectively.
In summary, \texttt{FREIGHT-con} produces the best connectivity among (buffered) streaming competitors, \hbox{outperforming even in-memory algorithm \texttt{HYPE}}.

\vspace*{-0.25cm}
\subparagraph*{Cut-Net.}
Next we examine at the cut-net metric.
In Figure~\ref{fig:cut_impr}, we plot the cut-net improvement over \texttt{Hashing}.
\texttt{PaToH-cut}  produces the best overall cut-net, with an average improvement of $100\%$ compared to \texttt{Hashing}.
\texttt{FREIGHT-cut} is found to be the second best algorithm with respect to cut-net, superior to internal-memory algorithm \texttt{HYPE}   and buffered streaming algorithm \texttt{RRS(0.15)}.
These three algorithms improve connectivity over \texttt{Hashing}  by $37\%$, $30\%$, and $17\%$ respectively.
Finally, both \texttt{MM-N2P} and \texttt{MM-L5} improve connectivity by $13\%$ on average over \texttt{Hashing}.
In direct comparison, \texttt{FREIGHT-cut} shows average connectivity improvements of $6\%$, $18\%$, $22\%$, and $22\%$ over \texttt{HYPE}, \texttt{RRS(0.15)}, \texttt{MM-N2P}, and \texttt{MM-L5}, respectively.
Each algorithm \hbox{preserves its relative ranking in average \NEW{cut-net} across all values of~$k$}.

In Figure~\ref{fig:cut_pp}, we plot cut-net performance profiles across all experiments.
In the plot, \texttt{PaToH-cut}  produces the best overall connectivity for $98.0\%$ of the instances, while \hbox{\texttt{FREIGHT-cut}} and \texttt{HYPE}   produce the best cut-net for $6.8\%$ and $5.2\%$ of the instances and all other streaming algorithms (\texttt{RRS(0.15)}, \texttt{MM-N2P}, \texttt{MM-L5}, and \texttt{Hashing}) produce the best cut-net for $4.8\%$ of the instances.
The cut-net results produced by  \texttt{FREIGHT-cut}, \texttt{HYPE}, \texttt{RRS(0.15)}, \texttt{MM-N2P}, \texttt{MM-L5}, and \texttt{Hashing}  are within a factor~$2$ of the best found cut-net for $83\%$, $79\%$, $69\%$, $66\%$, $66\%$, and $58\%$ of the instances, respectively.
This shows that \texttt{FREIGHT-cut} produces the best cut-net among all (buffered) streaming competitors and \hbox{even beats the in-memory algorithm \texttt{HYPE}}.

\vspace*{-0.25cm}
\subparagraph*{Running Time.}
Now we compare the algorithms' runtime.
Boxes and whiskers in Figure~\ref{fig:time_bp} display the distribution of the running time per pin, measured in nanoseconds, for all instances.
\texttt{Hashing}, \texttt{FREIGHT-cut}, and \texttt{FREIGHT-con} are the three fastest algorithms, with median runtimes per pin of \numprint{15}ns, \numprint{38}ns, and \numprint{41}ns, respectively.
\texttt{MM-L5}, \texttt{MM-N2P}, \texttt{HYPE}, and \texttt{RRS(0.15)}   follow with median runtimes per pin of \numprint{130}ns, \numprint{437}ns, \numprint{792}ns, and \numprint{833}ns, respectively.
Lastly, the algorithms with the \NEW{highest} median runtime per pin are \texttt{PaToH-cut}  and \texttt{PaToH-con}, with \numprint{2516}ns and \numprint{3333}ns respectively.
The measured runtime per pin for both \texttt{HYPE}   and \texttt{PaToH}  align with \hbox{values reported in prior research~\cite{schlag2020high}}.

In Figure~\ref{fig:time_pp}, we show running time performance profiles.
\texttt{Hashing}  is the fastest algorithm for $98.3\%$ of the instances, while \texttt{FREIGHT-cut} is the fastest one for $1.2\%$ of the instances and no other algorithm is the fastest one for more than $0.4\%$ of the instances. 
The running time of \texttt{FREIGHT-cut} and \texttt{FREIGHT-con} is within a factor 4 of that of \texttt{Hashing}  for $82\%$ and $72\%$ of instances, respectively. 
In contrast, for only $16\%$ of instances does this occur for \texttt{MM-L5}, and for less than $0.4\%$ of instances for all other algorithms.
The close running times of \texttt{FREIGHT} to \texttt{Hashing}  are surprising given \texttt{FREIGHT}'s superior solution quality compared to \texttt{Hashing}  and \hbox{all other streaming algorithms and even \texttt{HYPE}}.

\vspace*{-.25cm}
\subparagraph*{Further Comparisons.}
\label{subpar:Further Comparison}
For graph vertex partitioning \texttt{FREIGHT} and \texttt{Fennel} are mathematically equivalent. 
However, \texttt{FREIGHT} exhibits a lower computational complexity of $O(m+n)$ compared to the standard implementation of \texttt{Fennel}, which has a complexity of $O(m+nk)$ due to evaluating all blocks for each node. 
To optimize its performance for this use case, we have implemented an optimized version of \texttt{FREIGHT} with a memory consumption of $O(n+k)$, matching that of \texttt{Fennel}. 
In our experiments, we utilized the same graphs as in~\cite{StreamMultiSection} and tested with $k \in \{512,1024,1536,2048,2560\}$.
On average, \texttt{FREIGHT} proves to be 109 times faster than the standard implementation of \texttt{Fennel}. 
Moreover, the performance gap is found to increase as the value of $k$ grow, with \texttt{FREIGHT} reaching up to 261 times \hbox{faster than \texttt{Fennel} in some instances}.

\section{Conclusion}
\label{sec:conclusion}

In this work, we introduce \texttt{FREIGHT}, a highly efficient and effective streaming algorithm for hypergraph partitioning. 
Our algorithm leverages an optimized data structure, resulting in linear running time with respect to pin-count and linear memory consumption in relation to the numbers of nets and blocks.
The results of our extensive experimentation demonstrate that the running time of \texttt{FREIGHT} is competitive with the \texttt{Hashing}  algorithm, with a maximum difference of a factor of four observed in three fourths of the instances.
Importantly, our findings indicate that \texttt{FREIGHT} consistently outperforms all existing (buffered) streaming algorithms and even the in-memory algorithm \texttt{HYPE}, with regards to both cut-net and connectivity measures. 
This underscores the significance of our proposed algorithm as a highly efficient and effective solution for hypergraph partitioning \hbox{in the context of large-scale and dynamic data processing.}
Given our good results, we plan to publicly release our \hbox{algorithm soon}.

\bibliography{phdthesiscs}

\end{document}